\documentclass[preprint,12pt]{elsarticle}

\usepackage{amssymb}
\usepackage{amsmath}

\usepackage{amsthm}

\journal{Theoretical Computer Science}

\usepackage{lineno}
\usepackage{subcaption}

\usepackage{etoolbox}
\apptocmd{\sloppy}{\hbadness 10000\relax}{}{}

\usepackage{hyperref}
\hypersetup{colorlinks=true}

\newtheorem{theorem}{Theorem}
\newtheorem{corollary}[theorem]{Corollary}

\newtheorem{claim}[theorem]{Claim}
\newtheorem{observation}[theorem]{Observation}
\newtheorem{proposition}[theorem]{Proposition}

\newtheorem{lemma}[theorem]{Lemma}

\newtheorem{problem}{Problem}

\theoremstyle{definition}
\newtheorem{definition}[theorem]{Definition}

\begin{document}

\begin{frontmatter}



\title{On 1-Planar Graphs with Bounded Cop-Number \tnoteref{t1}} 
 \tnotetext[t1]{Research is partly funded by NSERC.}

\author[Carleton]{Prosenjit Bose} 

\author[uOttawa]{Jean-Lou De Carufel}

\author[Carleton]{Anil Maheshwari}

\author[Carleton]{Karthik Murali}

\affiliation[Carleton]{organization={School of Computer Science, Carleton University},
            city={Ottawa},
            country={Canada}}

\affiliation[uOttawa]{organization={School of Electrical Engineering and Computer Science, University of Ottawa},
            country={Canada}}

\begin{abstract}
Cops and Robbers is a type of pursuit-evasion game played on a graph where a set of cops try to capture a single robber. The cops first choose their initial vertex positions, and later the robber chooses a vertex. The cops and robbers make their moves in alternate turns: in the cops' turn, every cop can either choose to move to an adjacent vertex or stay on the same vertex, and likewise the robber in his turn. If the cops can capture the robber in a finite number of rounds, the cops win, otherwise the robber wins. The cop-number of a graph is the minimum number of cops required to catch a robber in the graph. It has long been known that graphs embedded on surfaces (such as planar graphs and toroidal graphs) have a small cop-number. Recently, Durocher et al. [Graph Drawing, 2023] investigated the problem of cop-number for the class of $1$-planar graphs, which are graphs that can be embedded in the plane such that each edge is crossed at most once. They showed that unlike planar graphs which require just three cops, 1-planar graphs have an unbounded cop-number. On the positive side, they showed that maximal 1-planar graphs require only three cops by crucially using the fact that the endpoints of every crossing in an embedded maximal 1-planar graph induce a $K_4$. In this paper, we show that the cop-number remains bounded even under the relaxed condition that the endpoints induce at least three edges. More precisely, let an $\times$-crossing of an embedded 1-planar graph be a crossing whose endpoints induce a matching; i.e., there is no edge connecting the endpoints apart from the crossing edges themselves. We show that any 1-planar graph that can be embedded without $\times$-crossings has cop-number at most 21. Moreover, any 1-planar graph that can be embedded with at most $\gamma$ $\times$-crossings has cop-number at most $\gamma + 21$. 
\end{abstract}




\end{frontmatter}



\section{Introduction}\label{sec: intro}

\textit{Pursuit-evasion games} are a class of mathematical problems where a set of pursuers attempt to track down evaders in an environment. They have been intensively studied due to their applications in robotics \cite{application_robotics}, network security \cite{application_network_security}, surveillance \cite{application_surveillance}, etc. \textit{Cops and Robbers} is a type of pursuit-evasion game played on a graph $G = (V,E)$ where a set $\mathcal{U}$ of cops try to capture a single robber. The game proceeds in \textit{rounds}, where each round consists first of the \textit{cops' turn} and then the \textit{robber's turn}. In the first round, the cops place themselves on some vertices of $G$, after which the robber chooses a vertex to place himself. In subsequent rounds, the cops' turn consists of each cop of $\mathcal{U}$ either moving to an adjacent vertex or staying on the same vertex. Likewise, the robber's turn consists of him either moving to an adjacent vertex or staying on the same vertex. At any point in time, more than one cop is allowed to occupy the same vertex. We assume that both the cops and the robber have complete knowledge of the environment and the moves that each player makes in their respective turns. The game terminates when the robber is \textit{captured}; this happens when the robber and some cop of $\mathcal{U}$ are both on the same vertex of $G$. If the robber can be captured in a finite number of rounds, the cops win, otherwise the robber wins.

\subsection{Brief Survey} The study of Cops and Robbers was initiated independently by Quilliot \cite{quilliot1978} and Nowakowski and Winkler \cite{nowakowski_winkler} in the context of a single cop. Aigner and Fromme \cite{AignerFromme} later generalised the setting to multiple cops, and introduced the concept of cop-number. The \textit{cop-number} of a graph $G$, denoted by $c(G)$, is the minimum number of cops required to capture the robber. A graph class $\mathcal{G}$ is \textit{cop-bounded} if the set $\{c(G): G \in \mathcal{G}\}$ is bounded; else $\mathcal{G}$ is \textit{cop-unbounded}. Examples of graph classes that are cop-bounded include chordal graphs \cite{nowakowski_winkler, quilliot1978}, planar graphs \cite{AignerFromme}, graphs of bounded genus \cite{quilliot1985_genus}, $H$-minor-free graphs \cite{Andreae_minor} and $H$-(subgraph)-free graphs \cite{joret}. On the other hand, examples of graph classes that are cop-unbounded include bipartite graphs \cite{bonato_book} and $d$-regular graphs \cite{Andreae_regular}. One of the deepest open problems on cop-number is Meyniel's conjecture which states that $c(G) \in O(\sqrt{n})$ for all graphs $G$. (See \cite{meyniel_survey} for a survey paper on Meyniel's conjecture and the book \cite{bonato_book} by Bonato and Nowakowski for an extensive introduction to Cops and Robbers.)

In this paper, we are interested in the cop-number of embedded graphs. Most results in this area are concerned with the cop-number of graphs with genus $g$. We first discuss planar graphs which have genus $g = 0$. Aigner and Fromme \cite{AignerFromme} proved that 3 cops are sufficient for any planar graph and that there are planar graphs that require 3 cops. Clarke \cite{clarke2002constrained} showed that outerplanar graphs have cop-number two. Bonato et al. \cite{characterisation_outerplanar} showed that an outerplanar graph has cop-number one if and only if it is chordal. However, a classification of which planar graphs have cop-numbers 1, 2 and 3 has not yet been found \cite{Topological_directions}. For graphs with genus $g > 0$, the best known bound on the cop-number is $\frac{4}{3}g + \frac{10}{3}$ \cite{genus_currentbest}. A long-standing conjecture by Schroeder \cite{schroder2001copnumber} is that $c(G) \leq g + 3$ for all values of $g$. By the discussion above, the conjecture holds true for planar graphs. It also holds for toroidal graphs where $g = 1$; in fact the cop-number of a toroidal graph is at most 3 \cite{Lehner21}.

The field of `beyond-planar' graphs has recently garnered significant interest in the graph drawing community (see \cite{didimo_survey} for a survey and \cite{hong2020beyond} for a book on the topic). A typical example is a \textit{$k$-planar graph} which is a graph that can be embedded on the plane such that each edge is crossed at most $k$ times. Recently, Durocher et al. \cite{optimal1plane} initiated the study of cop-number on 1-planar graphs. They show that unlike planar graphs, the class of 1-planar graphs is cop-unbounded. However, they show 3 cops are sufficient and sometimes necessary for a \textit{maximal 1-planar graph}: a 1-planar graph to which no edge can be added without violating 1-planarity. Another result in the same paper concerns \textit{outer-1-planar graphs}, which are graphs that can be drawn on the plane so that all vertices are on the outer-face and each edge is crossed at most once. These graphs have treewidth at most 3 \cite{auer_outer-1-planar}, and when combined with a result of Joret et al. \cite{joret} that the cop-number of a graph with treewidth $tw$ is at most $tw/2 + 1$, we get that the cop-number of outer-1-planar graphs is at most 2. Durocher et al. \cite{optimal1plane} show that an outer-1-planar graph has cop-number 1 if and only if it is chordal. One can extend the notion of outer-1-planarity to outer-$k$-planarity. Durocher et al. noted that an outer-$k$-planar graph has treewidth at most $3k + 11$ \cite{wood_outer-k-planar}, and by the result of Joret et al. \cite{joret}, an outer-$k$-planar graph has cop-number at most $(3k + 13)/2$.

In this paper, we generalise the cop strategy used for planar graphs to a class of 1-planar graphs obtained by enforcing conditions on the subgraph induced by the four endpoints of every crossing. Similar approaches have been taken in the study of other problems, such as Hamiltonicity and connectivity. In \cite{Hamiltonian_cycles}, Fabrici et al. study Hamiltonian paths and cycles in 1-planar graphs where the endpoints of every crossing induce $K_4$ (called locally maximal 1-planar graphs) and $K_4$ or $K_4 \setminus \{e\}$ (called weakly locally maximal 1-planar graphs). In \cite{Masters_thesis}, the equivalence between minimum vertex cuts and shortest separating cycles in planar graphs was extended to 1-planar graphs where the endpoints of every crossing induce a $K_4$, $K_4 \setminus \{e\}$ or $C_4$ (called bowtie 1-planar graphs). In \cite{biedl_murali}, Biedl and Murali introduced the concept of $\times$-crossings in 1-planar graphs: these are crossings where the endpoints do not induce any edge apart from the crossing pair of edges. They show that the vertex connectivity of 1-planar graphs embeddable without $\times$-crossings can be computed in linear time. The concept of $\times$-crossings has even been extended to $k$-planar graphs. In \cite{biedl2024parameterized}, Biedl, Bose and Murali show that the vertex connectivity of $k$-planar graphs embeddable without $\times$-crossings can be computed in $O(n)$ time (when $k$ is a constant).

\subsection{Our Results} To prove that maximal 1-planar graphs have cop-number at most 3, Durocher et al. \cite{optimal1plane} crucially use the fact that in an embedded maximal 1-planar graph, the endpoints of every crossing induce a $K_4$. In this paper, we show that any 1-planar graph that can be embedded on the plane without $\times$-crossings has cop-number at most 21. The cop-number remains bounded even if we allow for a constant number of $\times$-crossings. More precisely, any 1-planar graph that can be embedded with at most $\gamma$ $\times$-crossings has cop-number at most $\gamma + 21$.

Our proof is structured similar to the proof by Bonato and Nowakowski \cite{bonato_book} that planar graphs have cop-number at most 3. However, the presence of crossings brings with it many challenges. One of the main ideas in the proof for planar graphs is that a robber can be trapped in the interior of a cycle formed by the union of two shortest paths. The extension of this to 1-planar graphs poses difficulties because a robber can escape out of a cycle via edge crossings. (As shown in \cite{optimal1plane}, this problem can easily be circumvented if one assumes that the endpoints of every crossing induce a $K_4$.) Therefore, as a first step, we introduce the concept of guarding crossing points of a 1-planar graph $G$, and more generally, guarding subgraphs of $G^\times$, which is the plane graph formed by inserting dummy vertices at crossing points (refer to Section \ref{sec: prelims}). Then, we discuss how a set of five cops can guard the crossing points and vertices of $G$ on a shortest path of a 1-planar graph $G$ embeddable without $\times$-crossings (Section \ref{sec: shortest paths}). This enables us to guard paths and cycles of $G^\times$ in some special configurations, using which we can prove that the cop-number of $G$ is at most 21 (Section \ref{sec: guarding 1-plane graphs without x-crossings}). We conclude by showing that the cop-number remains bounded if we allow for a constant number of $\times$-crossings (Section \ref{sec: conclusion}).


\section{Preliminaries}\label{sec: prelims}

All graphs in this paper are undirected and loopless (although we allow parallel edges). An \textit{embedding} of a graph $G$ on the plane is an assignment of each vertex of $G$ onto a distinct point in $\mathbb{R}^2$ and each edge to a simple non-self-intersecting curve connecting the two vertices incident to the edge without intersecting any other vertices of $G$. We also assume that no two edges touch each other tangentially and that no three edges intersect at a point interior to the three curves. Let $G$ be a graph embedded on the plane. Two edges $e_1, e_2 \in E(G)$ \textit{cross} each other if there is a point on $c$ in $\mathbb{R}^2$ that is common to the interior of the curves representing $e_1$ and $e_2$ in the embedding, and in a sufficiently small disk around $c$, the two curves alternate in order (put simply, the two curves do not touch each other tangentially). The two edges $\{e_1,e_2\}$ form a \textit{crossing} in the embedding, and $c$ is the \textit{crossing point}. 

A graph is \textit{planar} if it has an embedding on the plane such that no two edges cross. Such an embedding is called a \textit{planar embedding}, and a planar graph with a planar embedding is called a \textit{plane graph}. Let $D$ be a planar embedding of a graph $G$. A \textit{face} of a plane graph is a region of $\mathbb{R}^2 \setminus D$. A \textit{1-planar graph} is a graph that can be embedded on the plane such that each edge is crossed at most once. Such an embedding is called a \textit{1-planar embedding}, and a 1-planar graph with a 1-planar embedding is called a \textit{1-plane graph}. Let $\{e_1,e_2\}$ be a crossing in a 1-plane graph $G$. The \textit{endpoints of the crossing} refer to the union of endpoints of $e_1$ and $e_2$. Any 1-plane graph can be drawn such that no two edges incident with the same vertex cross \cite{schaefer2012graph}; hence, we assume that there are four distinct endpoints for every crossing. Two endpoints of the crossing are \textit{opposite} if they are endpoints of the same edge, otherwise they are \textit{consecutive}. An $\times$-crossing of a 1-plane graph is a
crossing in which there is no edge connecting two consecutive endpoints of the crossing. We use the notation $\hat{\mathcal{G}}$ to denote the set of all 1-planar graphs that can be embedded without $\times$-crossings. As a matter of convention, whenever we talk of a graph $G \in \hat{\mathcal{G}}$, we assume that we are given a 1-planar embedding of $G$ without $\times$-crossings.

Let $G$ be a 1-plane graph. The \textit{planarisation} of $G$ is the plane graph obtained by inserting a \textit{dummy vertex} at every point where a pair of edges of $G$ cross each other. (In this paper, we will use the terms crossing points and dummy vertices synonymously.) For a subgraph $H \subseteq G$, we use the notation $H^\times$ to denote the sub-drawing of $H$ in $G^\times$. More precisely, $H^\times$ is the subgraph of $G^\times$ obtained by inserting a dummy vertex wherever an edge of $H$ is crossed by an edge of $G$ (not necessarily an edge of $H$). For example, if $e = (u,v)$ is an edge of $G$, then $e^\times = e$, if $e$ is uncrossed, else $e^\times$ is a path $(u,d,v)$ where $d$ is the crossing point at which $e$ is crossed. Note that the planarisation of $G$ is the same as the graph $G^\times$, but this may not be true for arbitrary subgraphs of $G$. In any graph $H^\times$, where $H \subseteq G$, a vertex $v \in V(H^\times)$ will be called a \textit{$G$-vertex} if $v \in V(G)$.

\subsection{Cops and Robbers on Embedded Graphs}

In Cops and Robbers, each player either moves from one vertex to an adjacent vertex of the graph or stays on the same vertex of the graph. Therefore, the presence of parallel edges does not affect the cop-number of the graph. We make use of this fact in the following way. Let $\{(u,v), (w,x)\}$ be a crossing in a graph $G \in \hat{\mathcal{G}}$. Let $c$ denote the crossing point in $G^\times$. Up to renaming, let $(u,w)$ be an edge connecting two consecutive endpoints of the crossing. If the three vertices $\{c,u,w\}$ do not bound a face of $G^\times$, then we add a parallel uncrossed edge connecting $u$ and $w$ so that $\{c,u,w\}$ now bounds a face. (This can always be done by drawing the edge close to the crossing point.) As adding parallel edges does not affect the cop-number, we may assume that we can always embed a graph $G \in \hat{\mathcal{G}}$ such that every crossing point and two consecutive endpoints of the crossing bound a face of $G^\times$. The uncrossed edge that connects the two consecutive endpoints will be called a \textit{kite edge at the crossing} (Figure \ref{fig: adding kite edges}).

\begin{figure}
    \centering
    \includegraphics[scale = 0.75, page = 1]{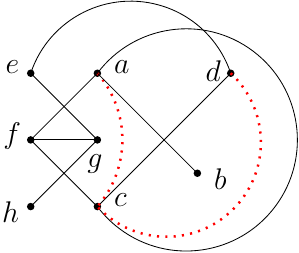}
    \caption{An example that illustrates the insertion of kite edges in a graph $G \in \hat{\mathcal{G}}$. The edges $(a,c)$ and $(c,d)$ are duplicated and inserted as uncrossed kite edges at the crossings $\{(a,b), (c,d)\}$ and $\{(a,c), (e,d)\}$ respectively. The kite edge $(f,g)$ is common to both the crossings $\{(e,g), (a,f)\}$ and $\{(g,h), (f,c)\}$.}
    \label{fig: adding kite edges}
\end{figure}

\paragraph{Guarding a subgraph of $G^\times$} Cops and robbers play on a graph $G$ by only using information about the adjacency of vertices, oblivious to how $G$ may be embedded on the plane. However, since our main tool to show that $\hat{\mathcal{G}}$ is cop-bounded relies on a specific embedding of the graph, we introduce the notion of guarding a subgraph of $G^\times$. We do this by first defining what we mean by guarding $G$-vertices and dummy vertices of $G^\times$. Let $G$ be a 1-plane graph and $v$ be a vertex of $G^\times$. If $v$ is a $G$-vertex, then $v$ is said to be \textit{guarded} if the robber cannot land on $v$ (because he gets captured before visiting $v$) or he is captured immediately after landing on $v$. (In the latter scenario, there is a cop already on $v$ or on a vertex adjacent to $v$.) If $v$ is a dummy vertex, then $v$ is said to be \textit{guarded} if the robber cannot move from any endpoint of the crossing to the opposite endpoint, or is captured immediately after he makes this move. (The robber may however be able to move between two consecutive endpoints of the crossing.) A graph $H \subseteq G^\times$ is said to be \textit{guarded} if all $G$-vertices and dummy vertices of $H$ are guarded. Observation \ref{obs: guarding a crossing point} shows a simple but useful way to guard a crossing point of $G$.

\begin{observation}\label{obs: guarding a crossing point}
    If two consecutive endpoints of a crossing are guarded, then the crossing point is guarded.
\end{observation}

\paragraph{Guarding in a subgraph of $G$} When we say that a subgraph $H \subseteq G^\times$ is guarded, we implicitly assume that $H$ is guarded no matter how the robber moves in the whole graph $G$. That is, the robber is free to move from any vertex to any other adjacent vertex of $G$, and no matter how he moves, $H$ remains guarded. Sometimes, we may also \textit{guard $H$ in a subgraph $I$ of $G$}. This means that $H$ is guarded so long as the robber is confined to $I$. (The cops need not be confined to $I$ however.) When the robber moves out of $I$ by taking an edge that does not belong to $E(I)$, then $H$ ceases to be guarded. Henceforth, whenever we say that a subgraph $H$ of $G^\times$ is guarded, we will always mean that $H$ is guarded in $G$ (i.e., the movement of the robber is not restricted), whereas, when we explicitly say that $H$ is guarded in a subgraph $I \subseteq G$, then $H$ is guarded so long as the robber is restricted to $I$.        

\medskip
When a subgraph $H \subseteq G^\times$ is guarded, the cops prevent the robber from using any edge $e$ of $G$ where $e^\times$ intersects $H$. This motivates the following definition of cop territory.

\begin{figure}
    \centering
    \includegraphics[scale = 0.75, page = 2]{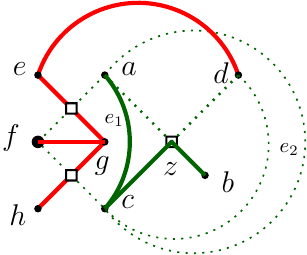}
    \caption{An illustration of cop and robber territory for the graph from Figure \protect\ref{fig: adding kite edges}. Let $H$ be the path $(a,c,z,b)$ of $G^\times$ that is guarded. The cop territory $\mathcal{C}(H)$ consists of edges $\{(a,b), (c,d), (f,a), (f,c), e_1, e_2\}$. The robber is on the vertex $f$ and the robber territory $\mathcal{R}(H)$ consists of the edges $\{(g,e), (g,f), (g,h), (e,d)\}$.}
    \label{fig: cop and robber territory}
\end{figure}

\begin{definition}[Cop Territory]\label{def: cop territory}
    Let $H$ be a subgraph of $G^\times$ guarded by a set $\mathcal{U}$ of cops. Then the \emph{cop territory}, denoted by $\mathcal{C}(H)$, is the subgraph of $G$ induced by all edges $e \in E(G)$ such that $e^\times \cap H \neq \emptyset$.
\end{definition}

Having defined cop territory, we will now define robber territory, which, informally speaking, is the subgraph of $G$ within which the robber can move without being captured. An important distinction between the two is that while cop territory was defined by a set of \textit{edges} prohibited for the robber, the robber territory will be defined by a set of \textit{vertices} that the robber can visit without being captured. Therefore, in Definition \ref{def: robber territory}, we define the robber territory by excluding edges of the cop territory, and not its vertices.

\begin{definition}[Robber Territory]\label{def: robber territory}
    Let $H$ be a subgraph of $G^\times$ guarded by a set $\mathcal{U}$ of cops. Then the \emph{robber territory}, denoted by $\mathcal{R}(H)$, is the connected component of $G \setminus E(\mathcal{C}(H))$ that contains the robber.
\end{definition}

(See Figure \ref{fig: cop and robber territory} for an illustration of cop territory and robber territory on the graph from Figure \ref{fig: adding kite edges}. All figures in the paper represent $G$-vertices with filled disks and dummy vertices with hollow squares.) The presence of kite edges at crossings always ensures that a robber territory has no $\times$-crossings, as Observation \ref{obs: R is 1-plane without x-crossings} shows.

\begin{observation}\label{obs: R is 1-plane without x-crossings}
    Let $H$ be a guarded subgraph of $G^\times$ for some $G \in \hat{\mathcal{G}}$. Then $\mathcal{R}(H)\in \hat{\mathcal{G}}$.
\end{observation}

\begin{proof}
    For every crossing in $\mathcal{R}(H)$, there is a kite edge at the crossing that belongs to $E(G)$. None of the four endpoints of such a crossing can belong to $H$, for the edges of the crossing belong to the robber territory. Therefore, the kite edge at the crossing, which is uncrossed by definition, belongs to $E(\mathcal{R})$. This implies the stated observation. 
\end{proof}

\begin{figure}
    \centering
    \includegraphics[scale = 0.75]{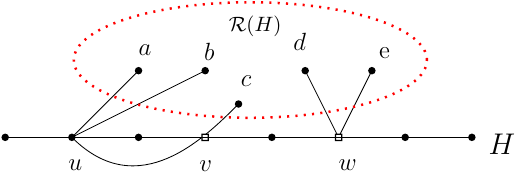}
    \caption{An example to illustrate `adjacency' to the robber territory. Here, $H$ is a path, and $S_H(u) = \{(u,a), (u,b)\}$ (even though $(u,c) \in E(G)$), $S_H(v) = \{(v,c)\}$ and $S_H(w) = \{(w,d), (w,e)\}$.} 
    \label{fig: adjacency}
\end{figure}

\paragraph{Adjacency to the robber territory} Let $H$ be a subgraph of $G^\times$ that is guarded. Since $H$ contains both $G$-vertices and dummy vertices, it is not straightforward to define when a vertex of $H$ is adjacent to a vertex of the robber territory, which is a subgraph of $G$. Therefore, we formalise a notion of adjacency below. For a $G$-vertex $v \in H$, let $S_H(v)$ be the set of all edges $e$ of $G$ incident with $v$ and an other vertex of $\mathcal{R}(H)$ such that $e^\times \cap H = \{v\}$.
For a dummy vertex $v \in H$, let $S_H(v)$ be the set of all edges $e$ of $G^\times$ incident with $v$ and an other vertex of $\mathcal{R}(H)$. For both cases above, we say that $v$ is \textit{adjacent} to $\mathcal{R}(H)$ or $v$ has a \textit{neighbour} in $\mathcal{R}(H)$ if $S_H(v) \neq \emptyset$. (See Figure \ref{fig: adjacency} for an illustration.) Below, we make an observation that will be useful in Section \ref{sec: guarding 1-plane graphs without x-crossings}. 

\begin{observation}\label{obs: adding a vertex to a 1-plane graph with no x-crossings}
     Let $H$ be a guarded subgraph of $G^\times$ for some $G \in \hat{\mathcal{G}}$. Let $\mathcal{R} := \mathcal{R}(H)$ be the robber territory. For any $v \in V(H)$, the graph $I := \mathcal{R} \cup S_H(v)$ belongs to $\hat{\mathcal{G}}$.
\end{observation}

\begin{proof}
   Consider any crossing of $I$. Since $G$ has no $\times$-crossings, there is a kite edge $e \in E(G)$ at this crossing. We will show that $e \in E(I)$, and therefore the crossing is not an $\times$-crossing in $I$. First, consider the case when $e$ is incident with $v$. Since all kite edges are uncrossed,  $e^\times \cap H = \{v\}$, and therefore $e \in S_H(v)$. This implies that $e \in E(I)$. On the other hand, if $e$ is not incident with $v$, then $e$ connects two vertices of $\mathcal{R}$. Again, since $e$ is uncrossed, we have $e \in E(I)$. Therefore, $I$ is a 1-plane graph without $\times$-crossings. 
\end{proof}


\section{Guarding Shortest Paths in 1-Planar Graphs}\label{sec: shortest paths}

One of the most useful results in Cops and Robbers is that a shortest path of a graph can be guarded by a single cop (after some finite number of initial rounds). This result, first obtained by Aigner and Fromme \cite{AignerFromme}, was used to show that the cop-number of planar graphs is at most three. One may wonder whether the same strategy could be used on the planarisation of a 1-plane graph $G$. Taking this approach however quickly lands us into many difficulties. For one, a shortest path in $G^\times$ between two $G$-vertices may not correspond to a shortest path in $G$. This happens because a shortest path in $G^\times$ may use a crossing point to move from one endpoint of the crossing to a consecutive endpoint. Clearly, such a move is not permissible for a cop. Furthermore, even if one assumes that $G \in \hat{\mathcal{G}}$, a cop must at least use the kite edge at the crossing to move between consecutive endpoints of the crossing. This adds a lag to the speed of the cop with respect to the robber, eventually rendering the path unguarded. 

To avoid running into the problems above, we take a different approach. We start with a shortest path $P$ in $G$, and show that if $G \in \hat{\mathcal{G}}$, then $P^\times$ can be guarded by a set of five cops. This result will later be used in Section \ref{sec: guarding 1-plane graphs without x-crossings} to show that $\hat{\mathcal{G}}$ is cop-bounded.

Let $\mathcal{U}$ be a set of cops. A strategy that we use to guard $P^\times$ is to make a single cop of $\mathcal{U}$ guard all $G$-vertices of $P^\times$, while on the other hand, all cops of $\mathcal{U}$ co-ordinate in guarding the dummy vertices of $P^\times$. This strategy will be used again in Section \ref{sec: guarding 1-plane graphs without x-crossings} to guard paths and cycles of $G^\times$. In light of this, we define the following notation: for any set of cops $\mathcal{U}$, let $\mathcal{U}^*$ denote a single `special' cop of $\mathcal{U}$.

\begin{lemma}\label{lem: five cops guard a path}
Let $G$ be a 1-plane graph and $P$ be a shortest path of $G$. Let $Q \subseteq V(P^\times)$ be the set of all $G$-vertices of $P^\times$ and all dummy vertices of $P^\times$ that are not part of $\times$-crossings in $G$. Then a set $\mathcal{U}$ of five cops can guard $Q$ such that $\mathcal{U}^*$ guards all $G$-vertices of $Q$. 
\end{lemma}

\begin{proof}
(See Figure \ref{fig: shortest path} for an illustration of Lemma \ref{lem: five cops guard a path}.) Since $P$ is a shortest path of $G$, all vertices of $P$ can be guarded by a single cop \cite{AignerFromme}. Hence, we can set up $\mathcal{U}$ such that at all times after a finite number of rounds $\alpha \geq 0$:

      \begin{enumerate}
        \item $\mathcal{U}^*$ guards all vertices of $P$.

        \item All vertices of $P$ within distance two of $\mathcal{U}^*$ are occupied by the remaining four cops.

        \item The robber is captured when he moves to a vertex adjacent to any of the five cops. 
    \end{enumerate}

We will show that $Q$ is guarded in all rounds starting from $\alpha + 1$. Trivially, all $G$-vertices of $Q$ are guarded by $\mathcal{U}^*$. Hence, we only need to show that the crossing points of $Q$ are guarded. For this, we will repeatedly use Observation \ref{obs: distance to path robber and primary cop}, which follows because $\mathcal{U}^*$ guards all vertices of $P$.

\begin{observation}\label{obs: distance to path robber and primary cop}
 At the end of the cops' turn of some round $\alpha' \geq \alpha$, if the distance of the robber to a vertex $v \in P$ is at most $k$, then the distance of $\mathcal{U}^*$ from $v$ is also at most $k$. 
\end{observation}

\begin{figure}
    \centering
    \includegraphics[scale = 0.75]{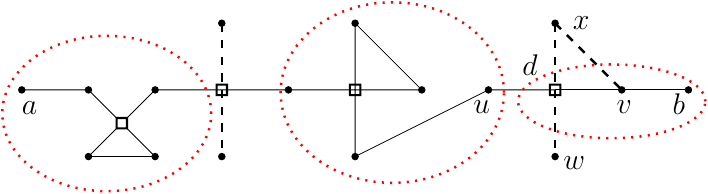}
    \caption{An illustration for Lemma \protect\ref{lem: five cops guard a path}. The solid edges form a shortest $(a,b)$-path $P$. The dashed edges belong to $G$, but not to $P$. The dotted region shows vertices of $P$ that belong to $Q$.}
    \label{fig: shortest path}
\end{figure}

Let $d$ be a dummy vertex on $Q$ resulting from an edge $(u,v)$ of $P$ being crossed by an edge of $G$, say $(w,x)$. We will show that if the robber moves from $w$ to $x$ in some round occurring after $\alpha$, then he is captured. In other words, if at the end of some round $\alpha' \geq \alpha$, the robber is on $w$, and at the end of round $\alpha'+1$, if he moved to $x$, then he is captured. We will first show that neither $(w,u)$ and $(w,v)$ can belong to $E(G)$. Suppose otherwise, for contradiction, and that $(w,u) \in E(G)$. Then, at the end of cops' turn of round $\alpha' + 1$, the cop $\mathcal{U}^*$ is on some vertex of $P$ that is within distance 1 from $u$ (Observation \ref{obs: distance to path robber and primary cop}). However, this implies that before the cops played their turn, there was some cop of $\mathcal{U}$ already on $u$ because starting from round $\alpha$, all vertices within distance two of $\mathcal{U}^*$ are occupied by some cop of $\mathcal{U}$. This means that the cop at $u$ must have captured the robber. This contradicts the fact that the robber's turn was played in round $\alpha' + 1$. Hence, $(w,u) \notin E(G)$ and $(w,v) \notin E(G)$. Since $d \in Q$, the crossing $\{(u,v), (w,x)\}$ has a kite edge, and hence at least one of $(x,u)$ or $(x,v)$ is in $E(G)$; by symmetry, we assume $(x,u) \in E(G)$. At the end of cops' turn of round $\alpha' + 1$, the cop $\mathcal{U}^*$ is on some vertex of $P$ within distance two from $u$ (Observation \ref{obs: distance to path robber and primary cop}). This implies that there is some cop on $u$ at the end of cops' turn. Thus, the robber is captured after it moves from $w$ to $x$ by the cop at $u$ in the round $\alpha'+2$.    
\end{proof}

Corollary \ref{cor: main corollary of shortest path} is a straightforward consequence of Lemma \ref{lem: five cops guard a path}.

\begin{corollary}\label{cor: main corollary of shortest path}
    Let $P$ be a shortest path of a graph $G \in \hat{\mathcal{G}}$. Then a set $\mathcal{U}$ of five cops can guard $P^\times$ such that $\mathcal{U}^*$ guards all $G$-vertices of $P^\times$.
\end{corollary}

With Corollary \ref{cor: main corollary of shortest path}, one may try to deduce, using the result by Aigner and Fromme for planar graphs \cite{AignerFromme}, that the cop number of a graph $G \in \hat{\mathcal{G}}$ is at most $3 \times 5 = 15$. This however fails because an essential strategy used to bound the cop-number for planar graphs is that a robber can be trapped on one side (either inside or outside) of the cycle formed by the union of two internally vertex disjoint shortest paths. One cannot extend this straight-forwardly to 1-planar graphs because a pair of vertex disjoint paths may cross each other repeatedly. With this, the notion of what constitutes inside and outside becomes non-trivial. In Section \ref{sec: guarding 1-plane graphs without x-crossings}, we give an alternative approach---although at a high-level, consisting of the same ideas as for planar graphs---to arrive at a bound on the cop number of a graph $G \in \hat{\mathcal{G}}$. 

Corollaries \ref{cor: stationary cops on shortest paths} and \ref{cor: five cops guard a path with x-crossings} are further implications of Lemma \ref{lem: five cops guard a path} that will be used in Section \ref{sec: guarding 1-plane graphs without x-crossings}.

\begin{corollary}\label{cor: stationary cops on shortest paths}
    Let $P$ be a shortest path in a 1-plane graph $G$ where at most $k$ edges of $P$ are involved in $\times$-crossings of $G$. Then a set $\mathcal{U}$ of $k+5$ cops can guard $P^\times$ such that $\mathcal{U}^*$ guards all $G$-vertices of $P^\times$.
\end{corollary}

\begin{proof}
     By Lemma \ref{lem: five cops guard a path}, we can make $\mathcal{U}^*$ guard all $G$-vertices of $P^\times$ while all cops of $\mathcal{U}$ together guard crossing points on those edges of $P$ that are not part of $\times$-crossings. For each edge $e \in E(P)$ that is involved in an $\times$-crossing with another edge $e' \in E(G)$, we place a cop on one endpoint of $e'$. This guards all crossing points of $P^\times$ on edges that are part of $\times$-crossings (Observation \ref{obs: guarding a crossing point}). 
\end{proof}

\begin{corollary}\label{cor: five cops guard a path with x-crossings}
Let $H$ be a guarded subgraph of $G^\times$ for some $G \in \hat{\mathcal{G}}$. Let $a$ and $b$ be two distinct $G$-vertices of $H$ that have a neighbour in $\mathcal{R} := \mathcal{R}(H)$. Let $I := \mathcal{R} \cup S_H(a) \cup S_H(b)$ and $P$ be a shortest $(a,b)$-path in $I$. Then a set $\mathcal{U}$ of five cops can guard $P^\times$ such that $\mathcal{U}^*$ guards all $G$-vertices of $P^\times$.
\end{corollary}

\begin{proof}
By Lemma \ref{lem: five cops guard a path}, we can make $\mathcal{U}^*$ guard all $G$-vertices of $P^\times$ while all cops of $\mathcal{U}$ together guard crossing points on those edges of $P$ that are not part of $\times$-crossings. In the graph $I$, all uncrossed edges incident with $a$ (resp. $b$) and some vertex of $\mathcal{R}$ belong to $S_H(a)$ (resp. $S_H(b)$). Therefore, any crossing of $I$ with at most one endpoint in $\{a, b\}$ has a kite edge that belongs to $I$. Hence, the only crossings of $I$ that are possibly $\times$-crossings are those where $a$ and $b$ are consecutive endpoints of the crossing. (The edge $(a,b)$ does not belong to $S_H(a) \cup S_H(b)$ since both $a,b \in V(H)$; therefore, $a$ and $b$ cannot be opposite endpoints of a crossing in $I$.)  However, since $P$ is a shortest $(a,b)$-path in $I$, both $a$ and $b$ are guarded. Therefore, by Observation \ref{obs: guarding a crossing point}, all dummy vertices of $\times$-crossings are guarded. 
\end{proof}


\section{Guarding a 1-Plane Graph Without \texorpdfstring{$\times$}{x}-Crossings}\label{sec: guarding 1-plane graphs without x-crossings}

The main objective of this section is to prove Theorem \ref{thm: main theorem}.

\begin{theorem}\label{thm: main theorem}
    The cop-number of any graph $G \in \hat{\mathcal{G}}$ is at most 21.
\end{theorem}

Our proof broadly follows the same structure as the proof in \cite{bonato_book} used to show that planar graphs have cop-number at most three. We maintain three sets of cops $\mathcal{U}_1$, $\mathcal{U}_2$ and $\mathcal{U}_3$ with seven cops in each set. The cops iteratively guard paths and cycles of $G^\times$ in some special ways, which we term as $\mathbb{P}$-Configuration and $\mathbb{C}$-Configuration respectively. 

\paragraph{$\mathbb{P}$-Configuration} One of the ways in which we will frequently guard a path of $G^\times$ is as follows.
First, we guard a subgraph $Q^\times \subseteq G^\times$ using a set $\mathcal{U}$ of five cops, where $Q^\times$ is derived from a shortest path $Q \subseteq G$ (Corollary \ref{cor: main corollary of shortest path}). Note that $Q^\times$ need not be a simple path in $G^\times$; that is, there can be a crossing $\{e,e'\}$ of $G$ such that both $e,e' \in Q$. Therefore, $Q^\times$ is a walk in $G^\times$. Let $\overline{Q}$ denote the simple path of $G^\times$ obtained from this walk by short-cutting at dummy vertices where $Q$ crosses itself. Since $Q^\times$ is guarded by $\mathcal{U}$, the subgraph $\overline{Q} \subseteq Q^\times$ is also guarded by $\mathcal{U}$. The cops $\mathcal{U}$ guard $\overline{Q}$ such that all $G$-vertices of $\overline{Q}$ are guarded by $\mathcal{U}^*$ (Corollary \ref{cor: main corollary of shortest path}). This motivates the following definition of $\mathbb{P}$-Configuration.

\begin{definition}[$\mathbb{P}$-Configuration]\label{def: path config}
Let $P$ be a path of $G^\times$ and $\mathcal{U}$ be a set of cops guarding $P$. Then $\mathcal{U}$ guards $P$ in \textit{$\mathbb{P}$-Configuration} if $\mathcal{U}^*$ guards all $G$-vertices of $P$.
\end{definition}

\paragraph{$\mathbb{C}$-Configuration} Another essential component of our proof is that of guarding cycles of $G^\times$ since this traps the robber within one side of the cycle, either outside or inside. We use the notation $D = [P_1, P_2]$ to denote a cycle $D \subseteq G^\times$ that is the union of two non-empty vertex-disjoint paths $P_1$ and $P_2$ of $G^\times$ together with two edges that connect endpoints of $P_1$ with endpoints of $P_2$. (For Definition \ref{def: cycle config}, recall the notation $\mathcal{C}(\cdot)$ for cop territory from Definition \ref{def: cop territory} and the notion of guarding in a subgraph of $G$ from Section \ref{sec: prelims}.)

\begin{figure}
    \centering
    \includegraphics[scale = 0.75]{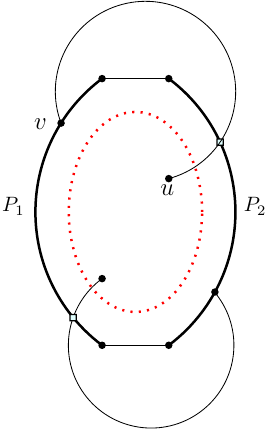}
    \caption{An illustration of a cycle $D = [P_1, P_2]$ guarded in $\mathbb{C}$-Configuration. The dotted region shows the robber territory $\mathcal{R}(D)$.}
    \label{fig: cycle configuration}
\end{figure}

\begin{definition}[$\mathbb{C}$-Configuration]\label{def: cycle config}

A cycle $D = [P_1, P_2]$ is guarded in $\mathbb{C}$-Configuration if there are disjoint sets of cops $\mathcal{U}_1$ and $\mathcal{U}_2$ such that: 

\begin{itemize}   
    \item $\mathcal{U}_1$ guards crossing points of $P_1$ and $\mathcal{U}_1^*$ guards all $G$-vertices of $P_1$ in $G \setminus E(\mathcal{C}(P_2))$.
    
    \item $\mathcal{U}_2$ guards crossing points of $P_2$ and $\mathcal{U}_2^*$ guards all $G$-vertices of $P_2$ in $G \setminus E(\mathcal{C}(P_1))$.
\end{itemize}
\end{definition}

We explain the definition through an example. Suppose that the robber moves from a vertex $u \in \mathcal{R}(D)$ to a vertex $v \in P_1$ (Figure \ref{fig: cycle configuration}). Let $e = (u,v)$ be the only edge of $G$ connecting $u$ and $v$. If $e^\times$ does not intersect with $P_2$, then $e \notin E(\mathcal{C}(P_2))$, and therefore, $\mathcal{U}_1^*$ captures the robber after it lands on $v$. If $e^\times$ intersects with $P_2$, then $e \in E(\mathcal{C}(P_2))$, and $\mathcal{U}_1^*$ may not capture the robber. However, by Definition \ref{def: cycle config}, the cops of $\mathcal{U}_2$ guard crossing points of $P_2$, implying that the robber is captured by $\mathcal{U}_2$ instead.


\paragraph{Invariants} Our strategy to prove Theorem \ref{thm: main theorem} is to progressively increase the number of guarded vertices of $G^\times$ across iterations until the whole graph $G^\times$ is guarded. An \textit{iteration}, indexed by an integer $\eta \geq 0$, is a sequence of consecutive rounds with the property that an iteration ends only when a set of four invariants $\mathcal{I}\ref{item: inv: path or cycle}$ - $\mathcal{I}\ref{item: inv: one cop free}$ hold. (The invariants will be stated shortly afterwards.) The iteration $\eta = 0$ is an empty set and iteration $\eta = 1$ begins with the first round of the game. More generally, every iteration $\eta > 1$ starts with the round occurring immediately after the last round of iteration $\eta - 1$.

For each iteration $\eta \geq 0$, let $L(\eta)$ represent some subgraph of $G^\times$ that is guarded. Initially, we set $L(0) = \emptyset$, and we ensure that if the current iteration is $\eta-1$, where $\eta \geq 1$, and the robber is not yet captured, then $L(\eta-1) \subsetneq L(\eta)$. At the end of each iteration, either a path $P$ will be guarded in $\mathbb{P}$-Configuration or a cycle $D$ will be guarded in $\mathbb{C}$-Configuration. These paths and cycles interface $L(\eta)$ with the robber territory $\mathcal{R}(L(\eta))$, i.e., no vertex of $L(\eta)$ apart from those in the path or cycle will have a neighbour in $\mathcal{R}(L(\eta))$. (Refer to Section \ref{sec: prelims} for the notion of adjacency to the robber territory.) Therefore, guarding these paths and cycles alone will be sufficient to guard all of $L(\eta)$. As we also require the cops to guard a path or a cycle every next iteration, we ensure that at the end of each iteration, at least one of $\mathcal{U}_1$, $\mathcal{U}_2$ or $\mathcal{U}_3$ is \textit{free}, meaning that they do not guard any subgraph of $G^\times$. We summarise these invariants below. At the end of each iteration $\eta \geq 1$:

\begin{enumerate}[({$\mathcal{I}$}1)]
    \item A path $P$ is guarded in $\mathbb{P}$-Configuration or a cycle $D$ is guarded in $\mathbb{C}$-Configuration. \label{item: inv: path or cycle}

    \item Let $J$ be the path $P$ or the cycle $D$ that is guarded. Then $J \subseteq L(\eta)$ and $L(\eta - 1) \subsetneq L(\eta)$.\label{item: inv: increase in guarded subgraph} 
    
    \item $\mathcal{R}(J) = \mathcal{R}(L(\eta))$, and no vertex of $L(\eta) \setminus J$ has a neighbour in $\mathcal{R}(J)$. \label{item: inv: interface}      

    \item At least one of $\mathcal{U}_1$, $\mathcal{U}_2$ or $\mathcal{U}_3$ is free. \label{item: inv: one cop free}
\end{enumerate}

For the first iteration, a set $\mathcal{U}_1$ of five cops guard a shortest path $Q$ in $G$. By this, $\overline{Q}$ is guarded in $G^\times$ in $\mathbb{P}$-Configuration, and we set $L(1) = \overline{Q}$. It is easy to verify that all the invariants above are satisfied. The configuration that the cops assume in subsequent iterations depends upon how many vertices of the path $P$ (in case of $\mathbb{P}$-Configuration) or the cycle $D$ (in case of $\mathbb{C}$-Configuration) are adjacent to the robber territory. In all cases, we will maintain the four invariants above.

\subsection{Cops in \texorpdfstring{$\mathbb{P}$}{P}-Configuration}\label{subsec: cops in Config-P}

Suppose that at the end of $\eta \geq 1$ iterations, a path $P \subseteq G^\times$ is guarded by a set $\mathcal{U}_1$ of cops in $\mathbb{P}$-Configuration. Depending upon how many vertices of $P$ have a neighbour in the robber territory $\mathcal{R} := \mathcal{R}(P)$, we consider different cases.

\subsubsection*{Case 1: Only one vertex of $P$ has a neighbour in $\mathcal{R}$}

\begin{figure}
     \centering
     \begin{subfigure}[b]{0.45\textwidth}
         \centering
         \includegraphics[scale = 0.75, page = 1]{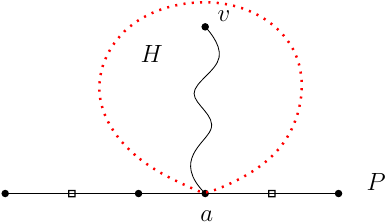}
         \subcaption{Case 1.1}
         \label{fig: case1.1}
     \end{subfigure}
     \hfill
     \begin{subfigure}[b]{0.45\textwidth}
         \centering
         \includegraphics[scale = 0.75, page = 2]{Path_1.pdf}
         \subcaption{Case 1.2}
         \label{fig: case1.2}
     \end{subfigure}
    \caption{An illustration of Case 1 of $\mathbb{P}$-Configuration. The vertices inside the dotted region denote the graph $H$.}
    \label{fig: case 1 of Path}
\end{figure}

Let $a \in P$ be the single vertex that has a neighbour in $\mathcal{R}$. Depending upon whether $a \in V(G)$ or not, we consider different subcases (Figure \ref{fig: case 1 of Path}). For all cases, we fix a $G$-vertex $v \in \mathcal{R}$.

\paragraph{Case 1.1: $a \in V(G)$} Since $a$ is the only vertex of $P$ that has a neighbour in the robber territory, $S_P(a) = \{(a,x): x \in \mathcal{R}\}$. In other words, all edges incident with $a$ and some vertex of the robber territory belong to $S_P(a)$. By Observation \ref{obs: adding a vertex to a 1-plane graph with no x-crossings}, the graph $H := \mathcal{R} \cup S_P(a)$ is a 1-plane graph without $\times$-crossings (Figure \ref{fig: case1.1}). By Corollary \ref{cor: main corollary of shortest path}, we can use five cops from $\mathcal{U}_2$ to guard a shortest $(a,v)$-path $P'$ in $H$. As a result, $\overline{P'} \subseteq G^\times$ is guarded in $\mathbb{P}$-Configuration. Let $Q = \overline{P'}$, and we free $\mathcal{U}_1$. 

\paragraph{Case 1.2: $a \notin V(G)$} Let $\{e_1,e_2\}$ be the crossing at $a$. Since $a$ has a neighbour in $\mathcal{R}$, there is an endpoint of the crossing, say $a_1$, that belongs to $\mathcal{R}$.
Let $a_2$ be another endpoint of the crossing consecutive with $a_1$. Consider the graph $H := \mathcal{R} \cup \{e_1,e_2\}$. Since $\mathcal{R}$ has no $\times$-crossings (Observation \ref{obs: R is 1-plane without x-crossings}), $H$ has at most one $\times$-crossing, possibly $\{e_1,e_2\}$ (Figure \ref{fig: case1.2}). First, we place a cop of $\mathcal{U}_2 \setminus \{\mathcal{U}_2^*\}$ at $a_2$. Now let $P'$ be a shortest $(a_1,v)$-path in $H$. By possibly using the additional cop at $a_2$, we can guard $(P')^\times \subseteq H^\times$ using six cops of $\mathcal{U}_2$ (Corollary \ref{cor: stationary cops on shortest paths}). If $a \in (P')^\times$, set $Q = \overline{P'}$, else set $Q := \overline{P'} \cup {(a,a_1)}$.
In either case, $a$ is guarded, and the path $Q$ is guarded by $\mathcal{U}_2$ in $\mathbb{P}$-Configuration. Finally, we free $\mathcal{U}_1$.

\medskip

After $Q$ has been guarded and $\mathcal{U}_1$ freed, we complete iteration $\eta + 1$ and set $L(\eta + 1) := L(\eta) \cup Q$. 

\begin{claim}
    At the end of iteration $\eta + 1$, all the four invariants hold.
\end{claim}

\begin{proof}
Invariants ($\mathcal{I}$\ref{item: inv: path or cycle}) and ($\mathcal{I}$\ref{item: inv: one cop free}) hold trivially. Since $a$ is the only vertex of $P$ that has a neighbour in $\mathcal{R}(P)$ and $a$ remains guarded throughout (because $a \in P \cap Q$), we have $\mathcal{R}(L(\eta + 1)) = \mathcal{R}(Q)$ and no vertex of $L(\eta+1) \setminus Q$ has a neighbour in $\mathcal{R}(Q)$.
Therefore, ($\mathcal{I}$\ref{item: inv: interface}) holds. Since $Q \subseteq L(\eta + 1)$ and $v \in Q \setminus L(\eta)$, we have $L(\eta) \subsetneq L(\eta + 1)$; hence ($\mathcal{I}$\ref{item: inv: increase in guarded subgraph}) holds.
\end{proof}

\subsubsection*{Case 2: More than one vertex of $P$ has a neighbour in $\mathcal{R}$}

Consider an enumeration of the vertices of $P$ in the order of their appearance on the path. Let $u$ and $v$ be the first and last vertices of $P$ that have a neighbour in $\mathcal{R}$. Let $a$ and $b$ be their (not necessarily distinct) neighbours in $\mathcal{R}$. Depending upon whether $u$ and $v$ are $G$-vertices or not, we consider different subcases.

\paragraph{Case 2.1: $u,v \in V(G)$} If $u,v \in V(G)$, then let $Q'$ be a shortest $(u,v)$-path in $H := \mathcal{R} \cup S_P(u) \cup S_P(v)$. By Corollary \ref{cor: five cops guard a path with x-crossings}, we can use five cops from $\mathcal{U}_2$ to guard $Q'$. Then $Q := \overline{Q'} \setminus \{u,v\}$ is guarded in $\mathbb{P}$-Configuration. 

\paragraph{Case 2.2: $u \in V(G)$ and $v \notin V(G)$ (or vice-versa).} Suppose that $u \in V(G)$ and $v \notin V(G)$. Let $Q'$ be a shortest $(u,b)$-path in $H := \mathcal{R} \cup S_P(u)$. By Observation \ref{obs: adding a vertex to a 1-plane graph with no x-crossings}, we can guard $Q'$ using five cops of $\mathcal{U}_2$. Then $Q := \overline{Q'} \setminus \{u\}$ is guarded in $\mathbb{P}$-Configuration. (The case where $u \notin V(G)$ and $v \in V(G)$ is symmetric.)

\paragraph{Case 2.3: $u, v \notin V(G)$.}
If both $u$ and $v$ are dummy vertices, then guard a shortest $(a,b)$-path $Q'$ in $H := \mathcal{R}$ using five cops from $\mathcal{U}_2$. Then $Q := \overline{Q'}$ is guarded in $\mathbb{P}$-Configuration. 

\medskip
Note that $\overline{Q'}$ is a path in the planarisation of $H$, and so is $Q$. This however does not imply that $Q$ is a path in $G^\times$. This is because there may exist an edge $e \in E(Q) \cap E(G)$ that is crossed by another edge $e'$ of $G$ incident with some vertex of $P$. (Since $e' \in \mathcal{C}(P)$, the edge $e$ may be uncrossed in $H$, and therefore also in $Q$ (Figure \ref{fig: path-2})). In this case, we subdivide the edge $e$ once to insert the dummy vertex. Having done this for all such edges, we get a path $P_2$ that is a subgraph of $G^\times$. Let $P_1$ be the sub-path of $P$ from $u$ to $v$. In Claim \ref{claim: P_1 and P_2 are vertex disjoint}, we show that $P_1$ and $P_2$ are non-empty paths forming a cycle.

\begin{claim}\label{claim: P_1 and P_2 are vertex disjoint}
$P_1$ and $P_2$ are non-empty and vertex disjoint paths of $G^\times$. Hence, $D = [P_1, P_2]$ is a simple cycle of $G^\times$.
\end{claim}

\begin{proof}
Since $u,v \in P_1$, we have $P_1 \neq \emptyset$. The path $P_2$ is not empty because $Q$ contains a vertex of $G^\times$ adjacent to $u$. (Note that this vertex cannot be $v$ because $S_P(u)$ and $S_P(v)$ do not contain the edge $(u,v)$.) It remains to show that $P_1$ and $P_2$ are vertex disjoint. For any edge $e \in S_P(u)$, we have that $e^\times \cap P = \{u\}$, and likewise for edges of $S_P(v)$. By the construction of $Q$ and $P_2$ above, we can infer that $P_2$ is vertex disjoint from $P_1$. Moreover, the endpoints of $P_1$ and $P_2$ can be joined by edges of $G^\times$ to obtain a simple cycle $D = [P_1, P_2]$ of $G^\times$.
\end{proof}

\begin{figure}
    \centering
    \includegraphics[scale = 0.85]{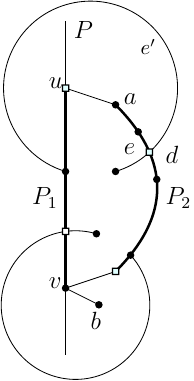}
    \caption{An illustration for Case 2 of $\mathbb{P}$-Configuration. }
    \label{fig: path-2}
\end{figure}

We now show that by some minor modifications, $\mathcal{U}_1$ and $\mathcal{U}_2$ can be made to guard $D$ in $\mathbb{C}$-Configuration. Note that the conditions necessary for $\mathbb{C}$-Configuration are already satisfied by the cops of $\mathcal{U}_1$ guarding $P_1$. In each of Cases 2.1 up to 2.3, the cops of $\mathcal{U}_2$ guard $Q$ in $H$ in $\mathbb{P}$-Configuration. Hence, $\mathcal{U}_2^*$ guards $G$-vertices of $P_2$ in $G \setminus E(\mathcal{C}(P))$. As $u$ and $v$ are the first and last vertices of $P$ to have a neighbour in $\mathcal{R}$, there is no edge $e \in E(\mathcal{C}(P))$ such that $e^\times$ intersects $P_2$ and $P \setminus P_1$. Therefore, $\mathcal{U}_2^*$ guards $G$-vertices of $P_2$ in $G \setminus E(\mathcal{C}(P_1))$. While $\mathcal{U}_2$ guards all dummy vertices of $Q$, there may exist dummy vertices on $P_2$ that are not dummy vertices on $Q$. (Recall from our earlier discussion that this happens because some edges of $Q$ may be crossed by edges that are incident to $P$.) Therefore, $\mathcal{U}_2$ may not guard all crossing points of $P_2$. To rectify this, we use one more cop of $\mathcal{U}_2$, and deploy it to imitate $\mathcal{U}_1^*$. In other words, this cop goes exactly in lockstep with $\mathcal{U}_1^*$. In Claim \ref{claim: lockstep tertiary cop path case}, we prove that this strategy works.

\begin{claim}\label{claim: lockstep tertiary cop path case}
    All dummy vertices of $P_2$ are guarded by $\mathcal{U}_2$.
\end{claim}

\begin{proof}
    We have already seen that all dummy vertices of $Q$ are guarded, hence we only consider dummy vertices of $P_2$ that do not belong to $Q$. Let $d$ be one such dummy vertex. Then the crossing at $d$ is of the form $\{e,e'\}$ where $e \in Q$ and $e'$ is incident to a vertex $x \in P$. We have seen that all $G$-vertices of $Q$ are guarded, hence the endpoints of $e$ are guarded. The vertex $x$ is guarded by $\mathcal{U}_1^*$ as $P$ is guarded in $\mathbb{P}$-Configuration. Since we have placed a cop of $\mathcal{U}_2$ in lockstep with $\mathcal{U}_1^*$, vertex $x$ is also guarded by $\mathcal{U}_2$. Therefore, by Observation \ref{obs: guarding a crossing point}, the cops of $\mathcal{U}_2$ guard $d$.  
\end{proof}


With this, the cops of $\mathcal{U}_2$ guard all crossing points of $P_2$, and $D$ is guarded in $\mathbb{C}$-Configuration. We end the current iteration $\eta + 1$. Since $D$ is guarded, $\mathcal{R}(D)$ must strictly lie inside or outside $D$. In either case, we set $L(\eta + 1) := L(\eta) \cup P_2$. 

\begin{claim}
    At the end of iteration $\eta+1$, all the four invariants hold.
\end{claim}

\begin{proof}
Trivially, ($\mathcal{I}$\ref{item: inv: path or cycle}) and ($\mathcal{I}$\ref{item: inv: one cop free}) are satisfied. As $u$ and $v$ are the first and last vertices on $P$ with a neighbour in $\mathcal{R}(P)$, any vertex of $P$ with a neighbour in $\mathcal{R}(P)$ belongs to $P_1$. Since $P_1$ remains guarded throughout (because $P_1 \subseteq P \cap D$),  we have $\mathcal{R}(L(\eta + 1)) = \mathcal{R}(D)$ and no vertex of $L(\eta+1) \setminus D$ has a neighbour in $\mathcal{R}(D)$. Hence ($\mathcal{I}$\ref{item: inv: interface}) is satisfied. We will show that ($\mathcal{I}$\ref{item: inv: increase in guarded subgraph}) also holds. Clearly $D \subseteq L(\eta + 1)$. Let $w$ be the vertex of $P_2$ adjacent to $u$. (By Claim \ref{claim: P_1 and P_2 are vertex disjoint}, $P_2$ is not empty.) If $w \in V(G)$, then $w \in \mathcal{R}(P)$ and $w \notin L(\eta)$. If $w \notin V(G)$, then $w$ is adjacent in $G^\times$ to a vertex $z$ of $\mathcal{R}(P)$ where $(u,z) \in E(G)$. By Invariant ($\mathcal{I}$\ref{item: inv: interface}) applied to iteration $\eta$, we have $w \notin L(\eta)$. Therefore, $L(\eta) \subsetneq L(\eta + 1)$.
\end{proof}

\subsection{Cops in \texorpdfstring{$\mathbb{C}$}{C}-Configuration}

Suppose that at the end of $\eta \geq 1$ iterations, a cycle $D \subseteq G^\times$ is guarded in $\mathbb{C}$-Configuration by cops $\mathcal{U}_1$ and $\mathcal{U}_2$. Hence, $\mathcal{R} := \mathcal{R}(D)$ is on one side of the cycle $D$; by symmetry, we may assume that it is in the interior of $D$. Depending upon how many vertices of $P_1$ and $P_2$ have adjacency in $\mathcal{R}$, we consider different cases.

\subsubsection*{Case 1: A single vertex of $D$ has a neighbour in $\mathcal{R}$}

In this case, we proceed exactly as we did for Case 1 of $\mathbb{P}$-Configuration. (The only minor change is that we substitute $S_D(v)$ for $S_P(v)$ everywhere.)

\subsubsection*{Case 2: A single vertex of $P_1$ and of $P_2$ have a neighbour in $\mathcal{R}$}

Let $a \in P_1$ and $b \in P_2$ be the two vertices that have a neighbour in $\mathcal{R}$. Let $a_1$ and $b_1$ be their (not necessarily distinct) neighbours in $\mathcal{R}$. Depending upon whether $a$ and $b$ are vertices of $G$ or not, we consider different subcases (Figure \ref{fig: case 2 of cycle}).

\begin{figure}
     \centering
     \begin{subfigure}[b]{0.3\textwidth}
         \centering
         \includegraphics[scale = 0.75, page = 1]{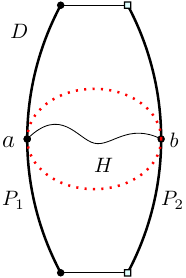}
         \subcaption{Case 2.1}
         \label{fig: case2.1}
     \end{subfigure}
     \hfill
     \begin{subfigure}[b]{0.3\textwidth}
         \centering
         \includegraphics[scale = 0.75, page = 2]{Cycle_2.pdf}
         \subcaption{Case 2.2}
         \label{fig: case2.2}
     \end{subfigure}
     \hfill
     \begin{subfigure}[b]{0.3\textwidth}
         \centering
         \includegraphics[scale = 0.75, page = 3]{Cycle_2.pdf}
         \subcaption{Case 2.3}
         \label{fig: case2.3}
     \end{subfigure}
    \caption{An illustration of Case 2 of $\mathbb{C}$-Configuration.}
    \label{fig: case 2 of cycle}
\end{figure}

\paragraph*{Case 2.1: $a,b \in V(G)$} If $a, b \in V(G)$ are the only vertices of $D$ that have a neighbour in the robber territory, then $S(a) = \{(a,x): x \in \mathcal{R}\}$ and $S(b) = \{(b,x): x \in \mathcal{R}\}$. In other words, all edges incident with $a$ (resp. $b$) and some vertex of the robber territory already belong to $S_D(a)$ (resp. $S_D(b)$). Let $H := \mathcal{R} \cup S_D(a) \cup S_D(b)$. By Corollary \ref{cor: five cops guard a path with x-crossings}, we can use five cops of $\mathcal{U}_3$ to guard a shortest $(a,b)$-path $P$ in $H$ (Figure \ref{fig: case2.1}). As a result, $a$ and $b$ are both guarded in $G$, and $\overline{P} \subseteq G^\times$ is guarded in $\mathbb{P}$-Configuration. Let $Q = \overline{P}$, and we free $\mathcal{U}_1 \cup \mathcal{U}_2$.

\paragraph*{Case 2.2: $a \in V(G)$ and $b \notin V(G)$ (or vice versa)} Suppose that $a \in V(G)$ and $b \notin V(G)$. Let $\{f_1,f_2\}$ be the crossing at $b$, and let $\{b_1,b_2\}$ be two consecutive endpoints of the crossing with $b_1 \in \mathcal{R}$. Note that $\{(a,x): x \in \mathcal{R}\} \subseteq S_D(a) \cup \{f_1,f_2\}$ since $a$ and $b$ are the only two vertices of $D$ with a neighbour in $\mathcal{R}$. Let $H := \mathcal{R} \cup S_D(a) \cup \{f_1,f_2\}$ (Figure \ref{fig: case2.2}). We first place a cop of $\mathcal{U}_3 \setminus \mathcal{U}_3^*$ at $b_2$. This cop will remain stationary. Since $\mathcal{R} \cup S_D(a)$ is 1-plane without $\times$-crossings (Observation \ref{obs: adding a vertex to a 1-plane graph with no x-crossings}), $H$ has at most one $\times$-crossing, possibly $\{f_1,f_2\}$. By Corollary \ref{cor: stationary cops on shortest paths}, we can guard a shortest $(a,b_1)$-path $P$ in $H$ using five cops and the cop at $b_2$. (As shown previously, all edges incident with $a$ and some vertex in $\mathcal{R}$ already belong to $H$; hence $a$ is guarded.) If $b \in P^\times$, set $Q := \overline{P}$; else set $Q := \overline{P} \cup (b_1,b)$. Since both $b_1$ and $b_2$ are guarded, so is $b$. This in turn implies that $Q$ is guarded, and we free $\mathcal{U}_1 \cup \mathcal{U}_2$. (The case where $a \notin V(G)$ and $b \in V(G)$ is symmetric.)

\paragraph*{Case 2.3: $a \notin V(G)$ and $b \notin V(G)$} Let $\{e_1,e_2\}$ and $\{f_1,f_2\}$ be the crossings at $a$ and $b$ respectively. Let $a_1,a_2$ and $b_1,b_2$ be two consecutive endpoints of the crossings at $a$ and $b$ respectively with both $a_1,b_1 \in \mathcal{R}$ (Figure \ref{fig: case2.3}). We place one cop of $\mathcal{U}_3 \setminus \mathcal{U}_3^*$ at $a_2$ and one more cop of $\mathcal{U}_3 \setminus \mathcal{U}_3^*$ at $b_2$, both of which remain stationary. Let $H := \mathcal{R} \cup \{e_1,e_2,f_1,f_2\}$. As $\mathcal{R}$ is 1-plane without $\times$-crossings (Observation \ref{obs: R is 1-plane without x-crossings}), $H$ has at most two $\times$-crossings, possibly $\{e_1,e_2\}$ and $\{f_1,f_2\}$. By Corollary \ref{cor: stationary cops on shortest paths}, we can guard a shortest $(a_1,b_1)$-path $P$ in $H$ using five cops and the two cops at $a_2$ and $b_2$. Both $a$ and $b$ are guarded because $a_1,a_2$ and $b_1, b_2$ are guarded. If $a \notin P^\times$, add $(a_1,a)$ to $\overline{P}$. Likewise, if $b \notin P^\times$, add $(b_1,b)$ to $\overline{P}$. Let $Q$ be this new path, and we free $\mathcal{U}_1 \cup \mathcal{U}_2$.

\medskip
After $Q$ has been guarded and $\mathcal{U}_1 \cup \mathcal{U}_2$ freed, we complete the iteration $\eta + 1$. We set $L(\eta + 1) := L(\eta) \cup Q$. We show that all the invariants hold. As before, ($\mathcal{I}$\ref{item: inv: path or cycle}) and ($\mathcal{I}$\ref{item: inv: one cop free}) hold trivially. Since the only vertices of $D$ with a neighbour in $\mathcal{R}(D)$ are $a$ and $b$ which remains guarded, $\mathcal{R}(L(\eta + 1)) = \mathcal{R}(D)$ and no vertex of $L(\eta+1) \setminus D$ has a neighbour in $\mathcal{R}(D)$. Hence ($\mathcal{I}$\ref{item: inv: interface}) holds. We will show that $L(\eta) \subsetneq L(\eta + 1)$. Let $w$ be the vertex of $Q$ adjacent to $a$. If $w \in V(G)$, then $w \in \mathcal{R}(D)$ and $w \notin L(\eta)$. If $w \notin V(G)$, then $w$ is adjacent in $G^\times$ to a vertex $z \in \mathcal{R}(D)$ where $(a,z) \in E(G)$. By Invariant ($\mathcal{I}$\ref{item: inv: interface}) applied to iteration $\eta$, we have $w \notin L(\eta)$. Therefore, $L(\eta) \subsetneq L(\eta + 1)$, and since $Q \subseteq L(\eta + 1)$, ($\mathcal{I}$\ref{item: inv: increase in guarded subgraph}) holds.

\subsubsection*{Case 3: $P_1$ or $P_2$ has more than one vertex with a neighbour in $\mathcal{R}$}

By symmetry, assume that $P_1$ has more than one vertex with a neighbour in $\mathcal{R}$. Enumerate the vertices of $P_1$, and let $u$ and $v$ be the first and last vertices of $P_1$ that have a neighbour in $\mathcal{R}$. Let $a$ and $b$ be their (not necessarily distinct) neighbours in $\mathcal{R}$. Just as in Case 2 of Section \ref{subsec: cops in Config-P}, we use $\mathcal{U}_3$ to guard a path $Q$ in a graph $H$, where the description of $H$ varies depending upon whether $u$ and $v$ are $G$-vertices or not. (A minor difference from Case 2 is that we replace $S_P(u)$ and $S_P(v)$ with $S_D(u)$ and $S_D(v)$ respectively.) 

Let $P_3$ be the path derived from $Q$ as follows. For every edge $e \in E(Q) \cap E(G)$ that is crossed by another edge of $G$ incident with some vertex of $D$, we subdivide the edge $e$ once to add the dummy vertex. Having done this for all such edges, we get a path $P_3$ that is a subgraph of $G^\times$ (Figure \ref{fig: cycle case 3}). Similar to Claim \ref{claim: P_1 and P_2 are vertex disjoint}, it is easy to show that $P_3$ is non-empty and vertex-disjoint from $D$. We use this fact to prove Claim \ref{claim: G-vertices of P_3 are guarded}.

\begin{figure}
    \centering
    \includegraphics{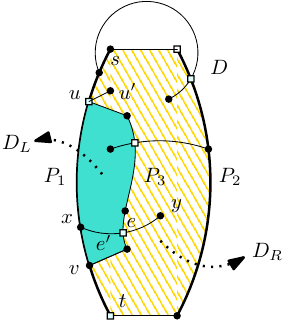}
    \caption{An illustration for Case 3 of $\mathbb{C}$-Configuration.}
    \label{fig: cycle case 3}
\end{figure}

\begin{claim}\label{claim: G-vertices of P_3 are guarded}
    All $G$-vertices of $P_3$ are guarded by $\mathcal{U}_3^*$.
\end{claim}

\begin{proof}
Note that $P_3$ is guarded in $\mathbb{P}$-Configuration in the graph $H$. Hence, all $G$-vertices of $P_3$ are guarded in $H$ by $\mathcal{U}_3^*$. We need to show that they are also guarded in $G$. As $P_3$ is vertex disjoint from $D$, and due to 1-planarity, there can be no edge $e \in E(G)$ with one endpoint in $\mathcal{R}$ and the other endpoint in $P_3$ such that $e^\times$ intersects $D$. Therefore, all $G$-vertices of $P_3$ are guarded by $\mathcal{U}_3^*$ in $G$. 
\end{proof}

As there may be dummy vertices on $P_3$ that are not dummy vertices of $Q$, the cops of $\mathcal{U}_3$ may not guard all crossing points of $P_3$. So, we deploy two hitherto unused cops of $\mathcal{U}_3$ so that one is in lockstep with $\mathcal{U}_1^*$ and the other is in lockstep with $\mathcal{U}_2^*$. In Claim \ref{claim: lockstep tertiary cop cycle case}, we prove that this strategy works.

\begin{claim}\label{claim: lockstep tertiary cop cycle case}
    All dummy vertices of $P_3$ are guarded by $\mathcal{U}_3$.
\end{claim}

\begin{proof}
    We have already seen that all dummy vertices of $Q$ are guarded, hence we only consider dummy vertices of $P_3$ that do not belong to $Q$. Let $d$ be one such dummy vertex. Then the crossing at $d$ is of the form $\{e,e'\}$ where $e \in Q$ and $e' = (x,y)$ with one endpoint, say $x$, being incident with $D$ (Figure \ref{fig: cycle case 3}). By symmetry, assume that $x \in P_1$. By Claim \ref{claim: G-vertices of P_3 are guarded}, the endpoints of $e$ are already guarded. To show that $d$ is guarded, it is sufficient to show that if the robber moves from $y$ to $x$, then it is captured by $\mathcal{U}_1^*$. (This is sufficient because there is a cop of $\mathcal{U}_3$ mimicking $\mathcal{U}_1^*$.) As $D$ is guarded in $\mathbb{C}$-Configuration, the cop $\mathcal{U}_1^*$ guards $x$ in $G \setminus E(\mathcal{C}(P_2))$. As $G$ is 1-plane, either $(e')^\times$ does not intersect $P_2$ or intersects $P_2$ at the $G$-vertex $y$. In the former case, $e' \notin E(\mathcal{C}(P_2))$ and in the latter case, the robber cannot land on $y$ without being captured. Therefore, the robber cannot move from $y$ to $x$ without being captured by $\mathcal{U}_1^*$. Hence, the claim is established. 
\end{proof}

Let $s$ and $t$ be the first and last vertices of $P_1$ (based on the same enumeration chosen at the start of this case). Let $P_1^-$ be the sub-path of $P_1$ from $u$ and $v$. Let $P_3^+$ be the path formed by the union of $P_3$ and the two sub-paths of $P_1$ from $s$ to $u$ and from $v$ to $t$. Let $D_L = [P_1^-, P_3]$ and $D_R = [P_3^+, P_2]$ (Figure \ref{fig: cycle case 3}). After $P_3$ has been guarded, the robber is confined either to the interior of $D_L$ or to the interior of $D_R$. We first consider the easy case when the robber is in the interior of $D_L$.

\paragraph{Case (a): Robber is in the interior of $D_L$} If the robber is in the interior of $D_L$, then we free $\mathcal{U}_2$. We will show that $\mathcal{U}_1$ and $\mathcal{U}_3$ guard $D_L$ in $\mathbb{C}$-Configuration. The conditions for $\mathbb{C}$-Configuration are already satisfied for $P_3$ due to Claims \ref{claim: G-vertices of P_3 are guarded} and \ref{claim: lockstep tertiary cop cycle case}. We will therefore focus on $P_1^-$. Since $D$ was guarded in $\mathbb{C}$-Configuration, the crossing points of $P_1^-$ are guarded by the cops of $\mathcal{U}_1$, and all $G$-vertices of $P_1^-$ are guarded by $\mathcal{U}_1^*$ in $G \setminus E(\mathcal{C}(P_2))$. We must now show the latter part also holds for $G \setminus E(\mathcal{C}(P_3))$. As $P_2$ lies to the exterior of $D_L$, and due to 1-planarity, there can be no edge $e \in E(G)$ with one endpoint in $\mathcal{R}(D_L)$ and the other endpoint in $P_1^-$ such that $e^\times$ intersects $P_2$ but not $P_3$. Therefore, $\mathcal{U}_1^*$ guards $G$-vertices of $P_1^-$ in $G \setminus E(\mathcal{C}(P_3))$. Hence, $D_L$ is guarded by $\mathcal{U}_1$ and $\mathcal{U}_3$ in $\mathbb{C}$-Configuration.

\paragraph{Case (b): Robber is in the interior of $D_R$} If the robber is in the interior of $D_R$, then we free $\mathcal{U}_1$. We will show that $\mathcal{U}_2$ and $\mathcal{U}_3$ guard $D_R$ in $\mathbb{C}$-Configuration. First, we show that the conditions for $\mathbb{C}$-Configuration are satisfied for $P_2$. Since $D$ was guarded in $\mathbb{C}$-Configuration, the crossing points of $P_2$ are guarded by $\mathcal{U}_2$. Moreover, all $G$-vertices of $P_2$ are guarded by $\mathcal{U}_2^*$ in $G \setminus E(\mathcal{C}(P_1))$. We now show that this also holds for $G \setminus E(\mathcal{C}(P_3^+))$. As $P_1^-$ lies to the exterior of $D_R$, and due to 1-planarity, there can be no edge $e \in E(G)$ with one endpoint in $\mathcal{R}(D_R)$ and the other endpoint on $P_2$ such that $e^\times$ intersects $P_1$ but not $P_3^+$. Therefore, $\mathcal{U}_2^*$ guards $G$-vertices of $P_2$ in $G \setminus E(\mathcal{C}(P_3^+))$. We are left with showing that the conditions for $\mathbb{C}$-Configuration are satisfied for $P_3^+$.

\begin{claim}\label{claim: robber in D_R}
    $\mathcal{U}_3$ guards crossing points of $P_3^+$ and $\mathcal{U}_3^*$ guards all $G$-vertices of $P_3^+$ in $G\setminus \mathcal{C}(P_2)$. 
\end{claim}

\begin{proof}    
By Claim \ref{claim: G-vertices of P_3 are guarded}, all $G$-vertices of $P_3$ are guarded by $\mathcal{U}_3^*$, and by Claim \ref{claim: lockstep tertiary cop cycle case}, all crossing points on $P_3$ are guarded by the cops of $\mathcal{U}_3$. Let $w$ be any vertex of $P_3^+ \setminus (P_3 \cup \{u,v\})$. As $u$ and $v$ are the first and last vertices of $P_1$ with a neighbour in $\mathcal{R}$, vertex $w$ is not adjacent to any vertex in $\mathcal{R}$. Therefore, if $w \in V(G)$, then any edge of $G$ incident with $\mathcal{R}(D_R)$ and $w$ must belong to $\mathcal{C}(P_2)$. And if $w$ is a dummy vertex, then there is no edge of $G^\times$ incident with $w$ and $\mathcal{R}(D_R)$. In the former scenario, $w$ is trivially guarded in $G \setminus \mathcal{C}(P_2)$, and in the latter scenario, $w$ is trivially guarded in $G$.

The above arguments show that we only need to restrict attention to $u$ and $v$. In what follows, we only consider the vertex $u$ since a symmetric argument will work for $v$. If $u$ is a $G$-vertex, then $P_3$ is derived from a path $Q' \subseteq H$ that contains $u$ (see Cases 2.1 and 2.2 of Section \ref{subsec: cops in Config-P}). Hence, the cop $\mathcal{U}_3^*$ guards $u$ in the graph $H$. As $u$ and $v$ are the first and last vertices of $P_1$ to have a neighbour in $\mathcal{R}(D)$, this further implies that $\mathcal{U}_3^*$ guards $u$ in $G \setminus E(\mathcal{C}(P_2))$. Next, consider the case when $u$ is a dummy vertex. All dummy vertices have exactly four neighbours, each corresponding to an endpoint at its crossing. Of the four neighbours of $u$, two of them belong to $D$ (one of which is on $P_1^-$ while the other can be on $P_1$ or $P_2$ depending upon the location of $u$ on $P_1$), and one belongs to $P_3$ (since $P_3$ is not empty). Therefore, there is at most one neighbour of $u$, say $u'$, that can belong to $\mathcal{R}(D_R)$. The endpoint of the crossing opposite to $u'$ must lie on $P_1^-$ due to the combinatorial embedding of the crossing (Figure \ref{fig: cycle case 3}). Put differently, the four half-edges at the crossing point alternate from belonging to one edge of the crossing to the other edge of the crossing. Hence, a robber that moves from $u'$ to its opposite endpoint lands on a $G$-vertex of $P_1^-$. The robber is then caught by the cop of $\mathcal{U}_3$ that imitates $\mathcal{U}_1^*$. 
\end{proof}


After $D_L$ or $D_R$ has been guarded and the cops of $\mathcal{U}_1$ or $\mathcal{U}_2$ have been freed, we end the current iteration $\eta + 1$. We set $L(\eta + 1) := L(\eta) \cup P_3$. Claim \ref{claim: invariants hold cycle-3} shows that all invariants hold.

\begin{claim}\label{claim: invariants hold cycle-3}
    At the end of iteration $\eta+1$, all the four invariants hold.
\end{claim}

\begin{proof}
    Firstly, Invariant ($\mathcal{I}$\ref{item: inv: path or cycle}) holds trivially. Invariant ($\mathcal{I}$\ref{item: inv: one cop free}) holds because at least one of $\mathcal{U}_1$ or $\mathcal{U}_2$ is freed. If the robber is in the interior of $D_L$, then $P_2$ has no neighbour in $\mathcal{R}(D_L)$, and if the robber is in the interior of $D_R$, then $P_1^-$ has no neighbour in $\mathcal{R}(D_R)$. Hence ($\mathcal{I}$\ref{item: inv: interface}) holds. We show that ($\mathcal{I}$\ref{item: inv: increase in guarded subgraph}) holds. Let $w$ be the vertex of $P_3$ adjacent to $u$. (One can show that $P_3$ is not empty by a proof similar to that of Claim \ref{claim: P_1 and P_2 are vertex disjoint}.) If $w \in V(G)$, then $w \in \mathcal{R}(D)$ and $w \notin L(\eta)$. If $w \notin V(G)$, then $w$ is adjacent in $G^\times$ to a vertex $z \in \mathcal{R}(D)$ such that $(u,z) \in E(G)$. By Invariant ($\mathcal{I}$\ref{item: inv: interface}) applied to iteration $\eta$, we have $w \notin L(\eta)$. Therefore, $L(\eta) \subsetneq L(\eta + 1)$, and since $D_L \cup D_R \subseteq L(\eta + 1)$, ($\mathcal{I}$\ref{item: inv: increase in guarded subgraph}) holds.
\end{proof}


\section{Concluding Remarks.}\label{sec: conclusion}

In this paper, we significantly widened the class of 1-planar graphs that were known to be cop-bounded by showing that not only maximal 1-planar graphs \cite{optimal1plane}, but any 1-planar graph embeddable without $\times$-crossings has a small cop-number (Theorem \ref{thm: main theorem}). Moreover, in Proposition \ref{prop: gamma + 21 cops}, we show that the class of 1-planar graphs that can be embedded with at most a constant number of $\times$-crossings is also cop-bounded.

\begin{proposition}\label{prop: gamma + 21 cops}
    Let $G$ be a 1-planar graph that can be embedded with at most $\gamma$ $\times$-crossings. Then $c(G) \leq \gamma + 21$.
\end{proposition}

\begin{proof}
    For each of the $\gamma$ $\times$-crossings, we place one cop on some endpoint of the crossing. Let $G'$ be the graph obtained by removing from each of the $\gamma$ crossings the crossed edge incident with the vertex containing the cop. By this, we get a graph $G' \in \hat{\mathcal{G}}$. It is easy to show that $c(G) \leq \gamma + c(G')$. We simply simulate the same cop strategy on $G$ as we had for $G'$. As $G$ is a super-graph of $G'$, the robber may make moves that are not possible in $G'$; however, any such move lands the robber on a vertex occupied by one of $\gamma$ cops. As $c(G') \leq 21$ by Theorem \ref{thm: main theorem}, we have $c(G) \leq \gamma + 21$.
\end{proof}

Theorem \ref{thm: main theorem} and Proposition \ref{prop: gamma + 21 cops} throw light on the fact that one of the possible reasons that 1-planar graphs have large cop-number is that such graphs cannot be embedded with few $\times$-crossings. However, our results do not imply that every 1-planar graph that requires many $\times$-crossings has a large cop-number. This leads to the following open problem.

\begin{problem}
    For every integer $k > 0$, classify which 1-planar graphs have cop-number at most $k$ and which 1-planar graphs have cop-number more than $k$.
\end{problem}

We have seen that one of the obstructions to generalising the cop strategy used for planar graphs to 1-planar graphs is the presence of $\times$-crossings. Incidentally, it was also shown in \cite{biedl_murali} that $\times$-crossings of 1-planar graphs pose the main challenge when generalising an algorithm for computing vertex connectivity from planar graphs to 1-planar graphs. It would be interesting to see for which other problems do $\times$-crossings obstruct such generalisations from planar graphs to 1-planar graphs.





\bibliographystyle{elsarticle-num-names} 
\bibliography{References}







\end{document}